\documentclass[reqno,12pt]{amsart}
\usepackage{amssymb,amsmath,amsthm,amscd,array}
\usepackage{cite}
\usepackage{mathdots}
\usepackage{dsfont}
\usepackage{graphicx}
\usepackage{float}
\usepackage{array}
\usepackage{subcaption}
\usepackage{bm}
\usepackage{thmtools}
\usepackage{caption}
\usepackage{subcaption}

\def\C{\mathbb{C}}

\def\P{\mathbb{P}}
\def\HH{\hat{\mathcal{H}}}
\def\ST{\hat{\mathcal{ST}}}
\def\SS{\mathcal{S}}
\def\TT{\hat{\mathcal{T}}}
\def\RR{\mathbb{R}}

\def\bp{{\bm \phi}}
\def\bt{{\bm \theta}}
\def\bw{{\bm \omega}}
\def\g{{\bf g}}
\def\X{{\bf X}}
\def\y{{\bf y}}
\def\Y{{\bf Y}}

\newcommand{\p}{{\bf p}}

\newcommand{\rr}{{\bf r}}
\newcommand{\vv}{{\bf v}}

\newcommand{\0}{{\bf 0}}
\newcommand{\1}{{\bf 1}}

\def\B{{\bf B}}
\def\BB{\hat{{\bf B}}}
\def\C{{\bf C}}
\def\D{{\bf D}}
\def\I{{\bf I}}
\def\J{{\bf J}}
\def\L{{\bf L}}
\def\P{{\bf P}}

\def\E{\hat{E}}
\def\Hh{{\hat{H}}}
\def\G{\hat{G}}
\def\R{\hat{R}}
\def\V{\hat{V}}
\def\T{\hat{T}}
\DeclareMathOperator{\vol}{Vol}
\def\up{\text{upper}}
\def\low{\text{lower}}
\def\stab{\text{stable}}
\def\lock{\text{locked}}

\newtheorem{theorem}{Theorem}
\newtheorem{lemma}[theorem]{Lemma}

\newtheorem{proposition}[theorem]{Proposition}

\theoremstyle{remark}
\newtheorem*{remark}{Remark}

\numberwithin{equation}{section}

\newtheorem{definition}{Definition}

\begin{document}


\title[Phase-locking in the Kuramoto Model with Asymmetric Coupling]{Volume Bounds for the Phase-locking Region in the Kuramoto Model with Asymmetric Coupling}
\author{Timothy Ferguson}

\begin{abstract}
The Kuramoto model is a system of nonlinear differential equations that models networks of coupled oscillators and is often used to study synchronization among them. It has been observed that if the natural frequencies of the oscillators are similar they will phase-lock, meaning that they oscillate at a common frequency with fixed phase differences. Conversely, we do not observe this behavior when the natural frequencies are very dissimilar. In \cite{doi:10.1137/16M110335X} Bronski and the author gave upper and lower bounds for the volume of the set of frequencies exhibiting phase-locking behavior. This was done under the assumption that any two oscillators affect each other with equal strength. In this paper the author generalizes these upper and lower bounds by removing this assumption. Similar to \cite{doi:10.1137/16M110335X} where the upper and lower bounds are sums over spanning trees of the network, our generalized upper and lower bounds are sums over certain directed subgraphs of the network. In particular, our lower bound is a sum over directed spanning trees. Finally, we numerically simulate the dependence of the number of directed spanning trees on the presence of certain two edge motifs in the network and compare this dependence with that of the synchronization of oscillator models found in \cite{10.3389/fncom.2011.00028}.
\end{abstract}

\maketitle




\section{Introduction} \label{sec:introduction}

The Kuramoto model was first defined in \cite{MR762432} by Y. Kuramoto in 1984 as a means to model the behavior of coupled oscillator networks. Applications of these oscillator networks include. For a review of synchronization phenomena as well as the history of the Kuramoto model see \cite{MR1783382}. The Kuramoto model on a network with $N$ ocillators is th system of $N$ coupled nonlinear differential equations
\begin{align} \label{eq:model}
\frac{d\theta_i}{dt} = \omega_i - \sum_{j \ne i} \gamma_{ij} \sin(\theta_i - \theta_j) \quad \text{for} \quad i \in \{1,\dots,N\}
\end{align}
where $\theta_i$ denotes the phase angle of the $i$th oscillator, $\omega_i$ the natural frequency of the $i$th oscillator, and $\gamma_{ij}$ the strength of the coupling allowing the $j$th oscillator to influence the $i$th oscillator. In many cases it is assumed that $\gamma_{ij} = \gamma_{ji}$ which physically means that the $j$th oscillator affects the $i$th oscillator in the same way as the $i$th oscillator affects the $j$th oscillator. In our previous paper \cite{doi:10.1137/16M110335X} Bronski and the author studied this special case, but in this paper we make no such assumption. In other words, we allow asymmetric coupling. This allows the model to be used to study a wider range of phenomena. For example, \cite{10.3389/fnhum.2010.00190} gives an extensive overview of the application of various versions of the Kuramoto model to studying cortical oscillation in neurobiology.

In particular, the Kuramoto model is used to study the phenomenon of synchronization (or phase-locking) whereby the network of oscillators rotate at a common frequency with fixed angle differences. In other words they rotate in unison as a whole rather than as individual oscillators. It has been observed that if the natural frequencies $\omega_i$ have similar values then the system synchronizes while it doesn't if their values are too far apart. It is then natural to study the set of natural frequencies for which the system can synchronize. One can also consider if there exists a system synchronizes to a stable state, namely a state to which the system returns if it is perturbed by a sufficiently small perturbation. As discussed in \cite{doi:10.1137/16M110335X} there is an extensive literature on this subject in the symmetric case when $\gamma_{ij} = \gamma_{ji}$. For example, see \cite{Dorfler.Chertkov.Bullo.2013,article44444}.

However, the asymmetric case, which we consider here, has received much less attention. In \cite{doi:10.1063/1.4954221} Skardal, Taylor, and Sun studied synchronization on directed networks and derived a function of the natural frequencies which can be used to optimize synchronization. Also, Restrepo, Ott, and Hunt \cite{articleyomama} studied directed networks under the assumption that each node is connected to many other nodes. They were able to approximate the onset of synchronization as well as the level of synchrony even with mixed positive/negative interactions. In \cite{7874076} Rao, Li, and Ogorzalek found sufficient conditions for the directed Kuramoto model to synchronize given a pacemaker, an oscillator to which all others are forced to rotate at the same rate and potentially achieve the same phase. Also, Shmidt, Papachristodoulou, M\"{u}nz, and Allg\"{o}wer \cite{SCHMIDT20123008} found several synchronization conditions for the directed Kuramoto model with delayed coupling.

As in \cite{doi:10.1137/16M110335X} we consider the problem of estimating the volume of the set of natural frequencies for which the Kuramoto model has a synchronous or stably synchronous solution. We do this as before by obtaining upper and lower bounds for these volumes as sums over spanning trees in the network. As a result we are able to deduce that these volumes can be well understood in terms of the number of spanning trees in the network.

The number of spanning trees in a network is known as the complexity of the network and is an important quantity for understanding its structure and properties. For example, the complexity of a network, along with the number of other classes of subgraphs, is directly related to the network's reliability \cite{Colbourn:1987:CNR:535891}. Unfortunately, it is well known that enumerating and/or counting spanning trees can be computationally expensive. Fortunately, due to the importance of this quantity, there is an extensive literature on this subject. Most of these results are either for random or networks or networks with uniform structure. For example, McKay \cite{MCKAY1983149} derived upper bounds for the complexity of $k$-regular networks. Later, Noga \cite{articlebilha,MCKAY1983149} found an asymptotic formula for the growth of the complexity of $k$-regular networks as $k$ tends to infinity. Similarly, Wu \cite{0305-4470-10-6-004} derived an asymptotic formula the complexity of lattices. \cite{articlebilha} spanning trees regular networks. Then Zhang, Wu, and Comellas \cite{ZHANG2014206} derived exact analytic formulas for the complexity of Apollonian networks. Apollonian networks are a type of small-world network related to Apollonian packing \cite{PhysRevLett.94.018702} and share structural properties with neuronal networks \cite{articlegivemeabreak}. Finally, we note that by Kirchoff's famous Matrix Tree Theorem the complexity of any network can be computed via sub determinants of the associated adjacency matrix.

There is also literature regarding the complexity of random networks. For example, Lyons, Peled, and Schramm \cite{articlenoway} study the asymptotic growth of the complexity of the largest connected component, known as the giant component, in an Erd\H{o}s-R\'{e}nyi random network. Also, in \cite{10.1007/978-3-319-50901-3_16}  Mokhlissi,  Lotfi, Debnath, and Marraki study the complexity of special classes of small-world networks. Small-world networks have been shown to provide good models for many real-world phenomena \cite{Amaral11149}. Also, in \cite{COHEN19879} Barlotti etal. studied the sensitivity of network complexity for random networks depending on certain parameters.

We now return to our discussion of synchronous (phase-locked) solutions of the Kuramoto model \eqref{eq:model}. Mathematically, a phase-locked solution is a solution of the form
\begin{align} \label{eq:phaselocked}
\bt(t) = \overline{\bw} t \1 + \bt_0
\end{align}
where $\overline{\bw}$ is the common frequency of the oscillators, $\1 = (1,\dots,1)^\top$ is the vector of all ones, and $\bt_0$ records the fixed phase angle differences. Plugging this into \eqref{eq:model} results in the equation
\begin{align} \label{eq:newphaselocked}
\overline{\bw} \1 = \bw - \g(\bt_0)
\end{align}
where $\bw$ is the vectors whose $i$th component is the natural frequency of the $i$th oscillator, $\omega_i$, and
\begin{align}
\g(\bt)_i := \sum_{j \ne i} \gamma_{ij} \sin(\theta_i - \theta_j) \quad \text{for} \quad i \in \{1,\dots,N\}.
\end{align}
Notice that with an appropriate choice of $\overline{\bw}$ we can suppose that $\bw$ is mean zero. Physically, this is equivalent to rotating the reference frame. Note that if $\gamma_{ij} = \gamma_{ji}$, then $\overline{\bw}$ is the mean of $\bw$. Therefore \eqref{eq:newphaselocked} states that $\bw$ has a phase-locked solution if and only if there exists a $\bt$ for which $\bw$ is the projection of $\g(\bt)$ onto the mean zero hyperplane $\1^\perp$. In other words, $\P \bw = \P \g(\bt)$ where $\P$ is an orthonormal projection onto $\1^\perp$. It should be emphasized that this is not equivalent to $\bw = \g(\bt)$ since the image of $\g$ is not generally constrained to $\1^\perp$ when the assumption $\gamma_{ij} = \gamma_{ji}$ is removed. See Figure \ref{fig:complex} for example. Furthermore, it is easy to see that this state is stable precisely when $\J$, the Jacobian of $\g$, is positive definite on $\1^\perp$. We restrict to $\1^\perp$ since $\g$ is invariant when $\bt$ is shifted by a multiple of $\1$. This motivates the definition of the following sets.

\begin{definition} \label{def:sets}
Define the sets
\begin{align} \label{eq:locked}
\Omega_\lock &:= \g(\Theta_\lock) \quad \text{where} \quad \Theta_\lock := \RR^N,
\end{align}
and
\begin{align} \label{eq:stable}
\Omega_\stab &:= \g(\Theta_\stab) \quad \text{where} \quad \Theta_\stab := \{ \bt \in \RR^N : \J(\bt) > 0 \text{ on } \1^\perp \}.
\end{align}
Clearly $\Theta_\stab \subseteq \Theta_\lock$ hence $\Omega_\stab \subseteq \Omega_\lock$.
\end{definition}

By the preceding discussion we then see that $\eqref{eq:model}$ has a phase-locked or stable phase-locked solution if and only if $\P \bw \in \P \Omega_\lock$ or $\P \bw \in \P \Omega_\stab$ respectively. However, as is demonstrated in Figure \ref{fig:complex}, the regions $\Omega_\lock$ and $\Omega_\stab$ as well as their projections $\P \Omega_\lock$ and $\P \Omega_\stab$ can be quite complex. As a result, we restrict ourselves to the problem of estimating their size. This motivates the following definition.

\begin{figure}[H]
\captionsetup{font=scriptsize}
\centering
\begin{subfigure}{.45\textwidth}
  \centering
  \includegraphics[width=.8\linewidth]{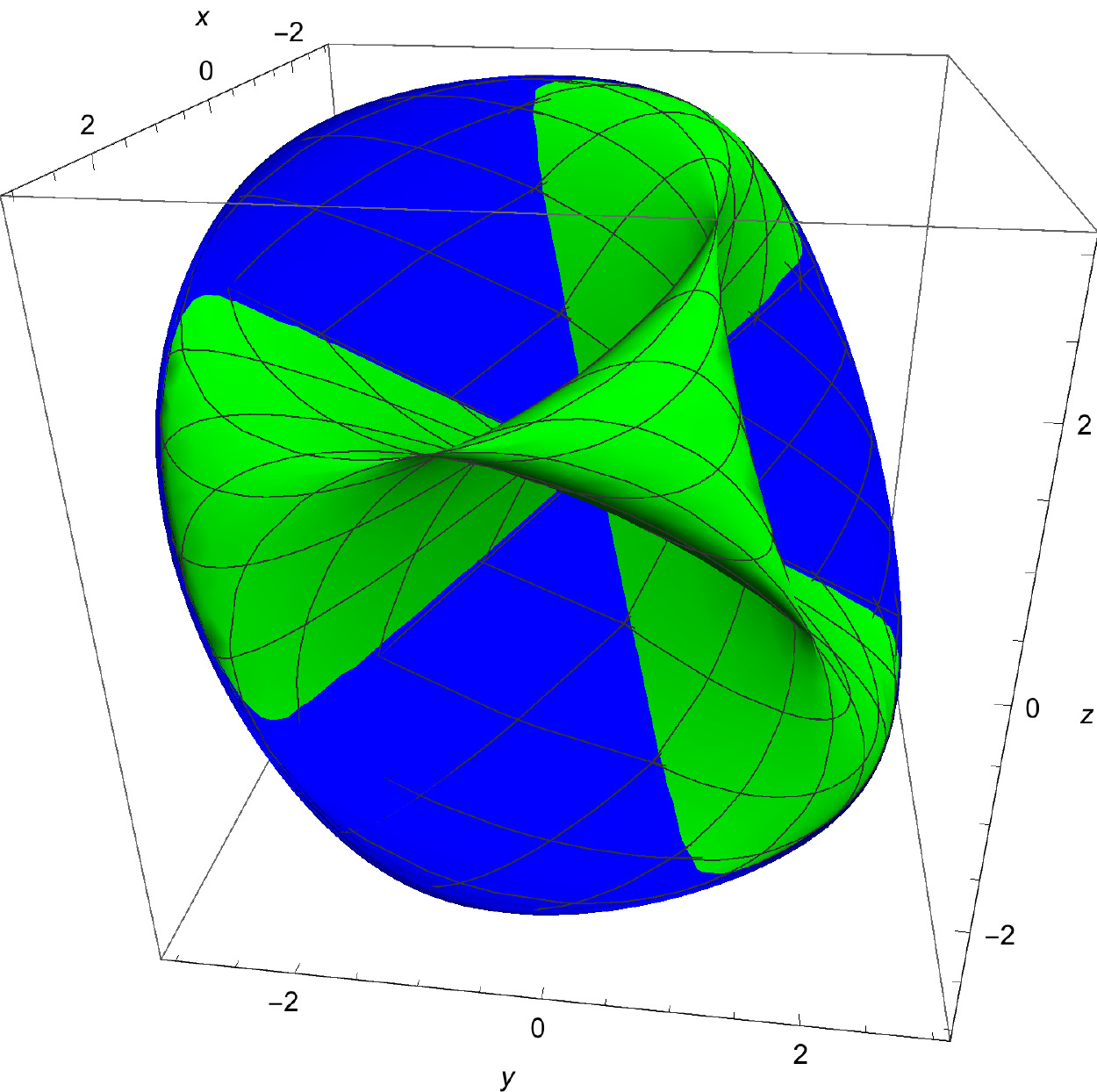}
\end{subfigure}
\begin{subfigure}{.45\textwidth}
  \centering
  \includegraphics[width=0.8\linewidth]{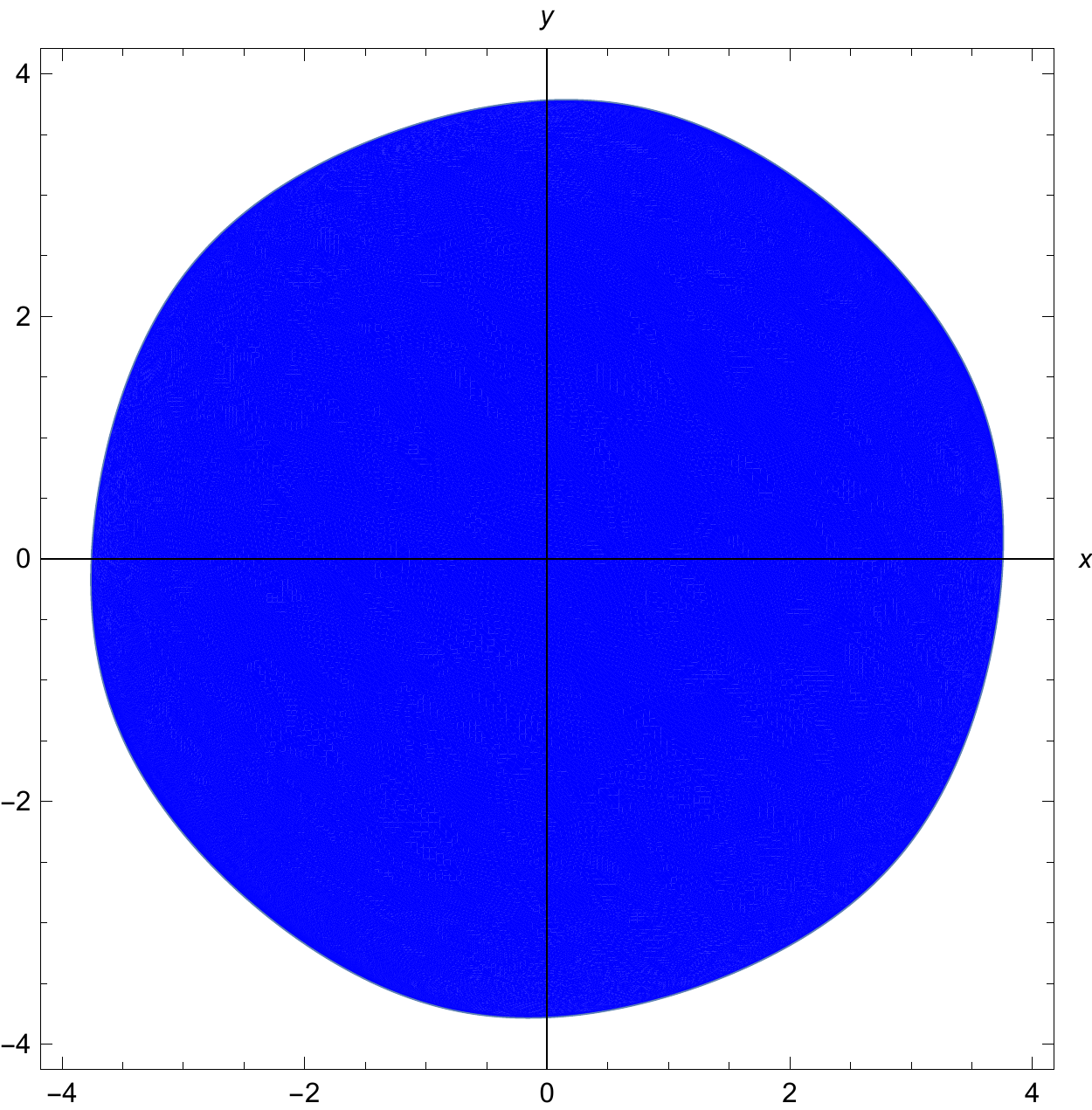}
\end{subfigure}
\caption{The left plot shows $\Omega_\lock$ for the complete graph on three vertices with asymmetric edge weights $\gamma_{12} = \gamma_{23} = \gamma_{31} = 1$ and $\gamma_{21} = \gamma_{32} = \gamma_{13} = 2$. The subset $\Omega_\stab$ is plotted as the blue subset. Note that neither of these sets lie completely in $\1^\perp$. The right plot is the projection of these sets onto $\1^\perp$, coincidentally the same, by means of the projection $\P$ with rows $\p_1 = (2,-1,-1)/\sqrt{6}$ and $\p_2 = (0,1,-1)/\sqrt{2}$.}
\label{fig:complex}
\end{figure}

\begin{definition} \label{def:volumes}
Define
\begin{align} \label{eq:volumes}
\vol_\lock := | \P \Omega_\lock| \quad \text{and} \quad \vol_\stab := | \P \Omega_\stab|
\end{align}
to be the $(N-1)$-dimensional Lebesgue measures of the projections of $\Omega_\lock$ and $\Omega_\stab$ onto the mean zero hyperplane.
\end{definition}

The main goal of this paper is to derive upper and lower bounds for $\vol_\lock$ and $\vol_\stab$. In particular, we will define quantities $\vol_\up$, $\vol_\up'$, and $\vol_\low$ such that
\begin{align} \label{eq:inequality}
\vol_\low \le \vol_\stab \le \vol_\lock \le \min \{ \vol_\up, \vol_\up' \}.
\end{align}
In Theorem \ref{thm:main} we will obtain formulas for these quantities as sums over certain subgraphs of the directed network $\G$ which we now define.

\begin{definition} \label{def:T}
A connected subgraph $\T$ of $\G$ is called a directed tree if it has a vertex, called the root, which has out degree zero and if every other vertex in $\T$ has out degree one. We let $\TT$ denote the set of all directed trees in $\G$ and $\TT_i$ those with root $i$. If in addition $\T$ contains every vertex of $\G$, it is called a directed spanning tree. We let $\ST$ denote the set of all directed spanning trees of $\G$ and $\SS\TT_i$ those with root $i$. For any directed tree we define the weight
\begin{align*}
\gamma(\T) : = \prod_{e \in \E_{\T}} \gamma_e.
\end{align*}
\end{definition}

\begin{definition} \label{def:H}
Let $\Hh$ denote a subgraph of $\hat{G}$ of the form
\begin{align*}
\Hh = \T \sqcup \bigsqcup_{\R} \R \bowtie_r \T_r.
\end{align*}
This represents a graph which is the disjoint union of a directed tree $\T$ with graphs of the form $\R \bowtie_r \T_r$ which are directed trees $\T_r$ rooted and attached to $\R$ at $r$. Here $\R$ can denote one of two different types of graphs. First, it can denote a single directed cycle, where all edges point in the same direction around the cycle, in which case we define the weight
\begin{align*}
\gamma(\R) := \prod_{e \in \E_{\R}} \gamma_e.
\end{align*}
Second, it can denote a single directed cycle $\R_+$ along with the reverse directed cycle $\R_-$, same ``edges" but in the opposite direction, forming a double directed cycle in which case we define the weight
\begin{align*}
\gamma(\R) := |\gamma(\R_+) - \gamma(\R_-)|.
\end{align*}
(Note that we don't allow $\R$ to be a single directed cycle if it is a subgraph of a double directed cycle.) We let $\HH$ denote the set of all such graphs $\Hh$ and $\HH_i$ those for which $\T$ has root $i$. For any such graph we define the weight
\begin{align*}
\gamma(\Hh) = \gamma(\T) \prod_{\R} \gamma(\R) \prod_r \gamma(\T_r).
\end{align*}
\end{definition}

\begin{theorem} \label{thm:main}
For any directed graph $\hat{G}$ with a directed spanning tree,
\begin{align}
\vol_\up := \frac{2^{N-1}}{\sqrt{N}} \sum_{\Hh \in \HH} \gamma(\Hh) \quad \text{and} \quad \vol_\up' := \frac{(2\pi)^{N-1}}{\sqrt{N}} \sum_{\T \in \ST} \gamma(\T)
\end{align}
and
\begin{align}
\vol_\low := \frac{1}{\sqrt{N}} \sum_{\T \in \ST} \gamma(\T) \sum_{i=1}^N \prod_{j \ne i} I(\deg_j(\T))
\end{align}
satisfify \eqref{eq:inequality} where $I(x) := \frac{1}{2} \int_0^{\pi/2} (\cos \theta + \sin \theta)^x d\theta$ and $\deg_j(\T)$ is the total degree of the $j$th vertex in the directed spanning tree $\T$, namely, the sum of its in and out degrees.
\end{theorem}

\begin{remark}
Theorem \ref{thm:main} generalizes to lower bounds for the volumes $\vol_\lock := | \P_\p \Omega_\lock|$ and $\vol_\stab := | \P_\p \Omega_\stab|$ where $\P_\p$ is an orthonormal projection onto the hyperplane with normal vector $\p$ all of whose components have the same sign.
\end{remark}

In Section \ref{sec:notation} we introduce notation and preliminary results which are used in the proof of Theorem \ref{thm:main} in Section \ref{sec:bounds}. In Section \ref{sec:comp} we compare our results with those obtained by Bronski and the author \cite{doi:10.1137/16M110335X} for undirected networks. Finally, in Section \ref{sec:motifs} we give a numerical example of how the complexity of a network, the number of spanning trees, depends in the statistical properties of certain two edge subgraphs called motifs. This is motivated by the work of Nykamp, Zhao, etal. \cite{10.3389/fncom.2011.00028,Zhao.2012} who studied the dependence of synchronization on the statistical properties of these motifs.



\section{Notation and Preliminaries} \label{sec:notation}

In this section we state all preliminary results as well as establish notation. We start with the underlying network.

\begin{definition} \label{def:directed}
Let $\G = (V,\E,\Gamma)$ be a weighted directed graph with vertex set $V = \{1,\dots,N\}$, edge set $\E \subseteq \{ (i,j) : i \ne j \in V\}$, and positive edge weights $\Gamma = \{ \gamma_e \}_{e \in E}$. The edge $e=(i,j)$ represents a directed edge from $i$ to $j$, and in this case, we set $\gamma_{ij} = \gamma_e$.
\end{definition}

\begin{definition} \label{def:undirected}
Given a graph $\G$ as in Definition \ref{def:directed}, define the undirected graph $G = (V,E)$ with vertex set $V = \V$ and edge set
\begin{align*}
E  = \{ \{i,j \} : (i,j) \in \E \text{ or } (j,i) \in \E \}.
\end{align*}
In general, given a subgraph $\Hh$ of $\G$ define the subgraph $H$ of $G$ by $V_H = \V_\Hh$ and
\begin{align*}
E_H  = \{ \{i,j \} : (i,j) \in \E_\Hh \text{ or } (j,i) \in \E_\Hh \}.
\end{align*}
Furthermore we fix a labeling and orientation of the edges of $G$.
\end{definition}

We give an example of such a $\G$ and its corresponding $G$ in Figure \ref{fig:graphs}.

\begin{figure}[H]
\captionsetup{font=scriptsize}
\centering
\begin{subfigure}{.45\textwidth}
  \centering
  \includegraphics[width=.8\linewidth]{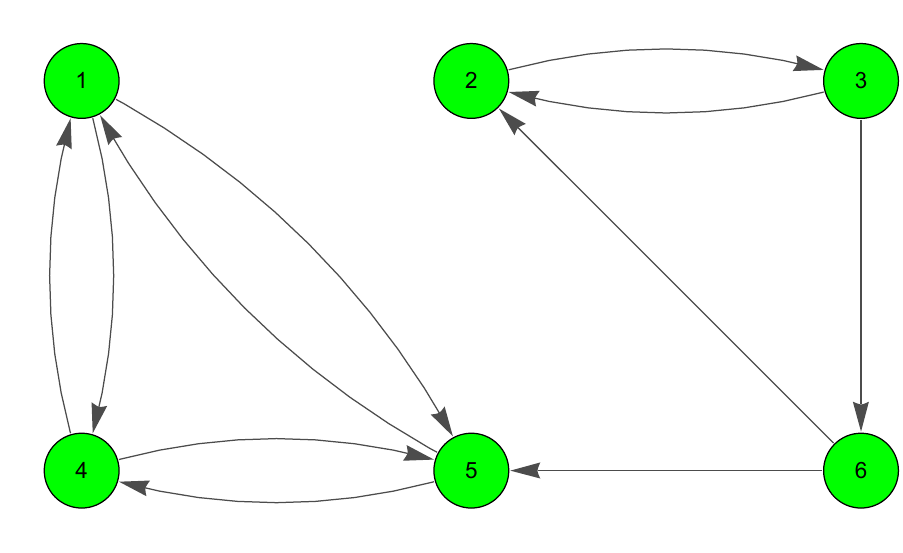}
\end{subfigure}
\begin{subfigure}{.45\textwidth}
  \centering
  \includegraphics[width=0.8\linewidth]{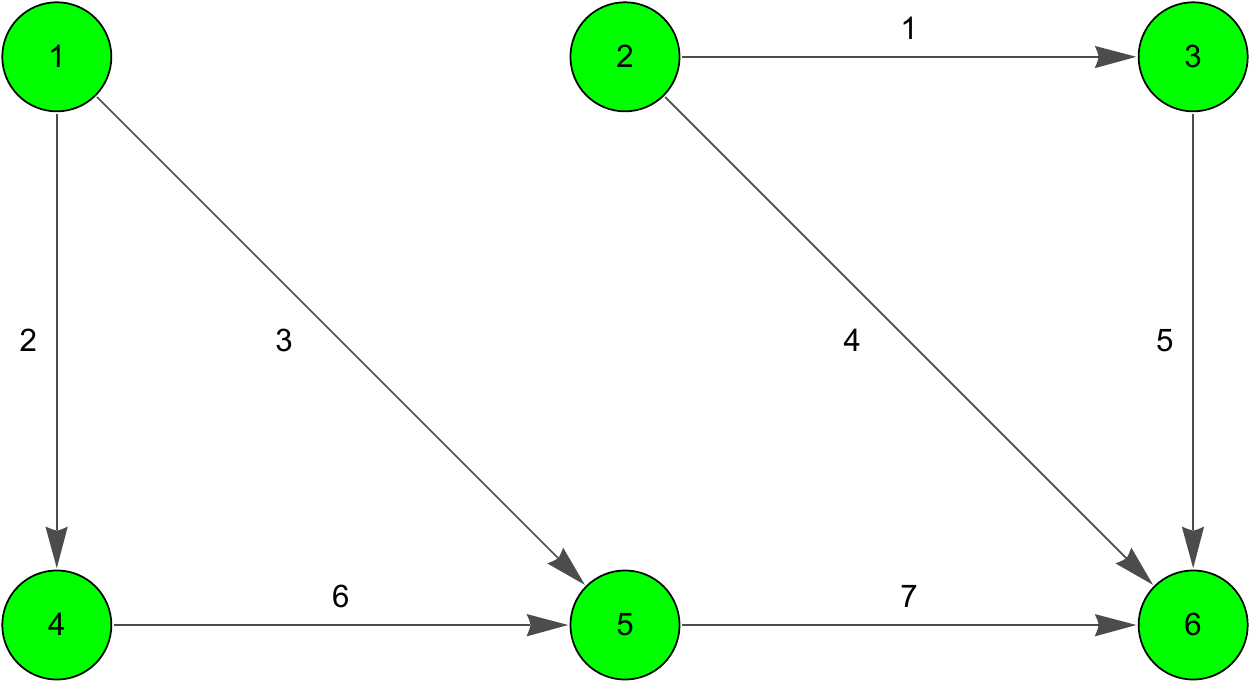}
\end{subfigure}
\caption{A graph $\G$ and its corresponding $G$ with labeled edges and orientation.}
\label{fig:graphs}
\end{figure}

\begin{definition} \label{def:incidence}
Given graphs $\G$ and $G$ as in Definitions \ref{def:directed} and \ref{def:undirected}, define the two $|V| \times |E|$ incidence matrices $\BB$ and $\B$ by
\begin{align*}
\BB_{ie} =
\begin{cases}
\gamma_{ij} & \mbox{if $e = (j,i)$,}
\\
-\gamma_{ij} & \mbox{if $e = (i,j)$,}
\\
0 & \mbox{otherwise,}
\end{cases}
\quad \text{and} \quad
\B_{ie} =
\begin{cases}
1 & \mbox{if $e = (*,i)$,}
\\
-1 & \mbox{if $e=(i,*)$,}
\\
0 & \mbox{otherwise.}
\end{cases}
\end{align*}
\end{definition}
Note that in the symmetric case, $\gamma_{ij} = \gamma_{ji}$ that $\BB = \B \D$ where $\D$ is the diagonal matrix with edge weights along the diagonal. For the graphs $G$ and $\G$ in Figure \ref{fig:graphs} we get that
\begin{align*}
\BB =
\begin{pmatrix}
0 & -\gamma_{14} & -\gamma_{15} & 0 & 0 & 0 & 0
\\
-\gamma_{23} & 0 & 0 & 0 & 0 & 0 & 0
\\
\gamma_{32} & 0 & 0 & 0 & -\gamma_{36} & 0 & 0
\\
0 &  \gamma_{41} & 0 & 0 & 0 & -\gamma_{45} & 0
\\
0 & 0 & \gamma_{51} & 0 & 0 & \gamma_{54} & 0
\\
0 & 0 & 0 & \gamma_{62} & 0 & 0 & \gamma_{65}
\end{pmatrix}
\quad \text{and} \quad
\B =
\begin{pmatrix}
0 & -1 & -1 & 0 & 0 & 0 & 0
\\
-1 & 0 & 0 & -1 & 0 & 0 & 0
\\
1 & 0 & 0 & 0 & -1 & 0 & 0
\\
0 & 1 & 0 & 0 & 0 & -1 & 0
\\
0 & 0 & 1 & 0 & 0 & 1 & -1
\\
0 & 0 & 0 & 1 & 1 & 0 & 1
\end{pmatrix}.
\end{align*}
Note that these matrices depend on the orientation and labeling of the edges and vertices of $G$ which is arbitrary. However, we always have that
\begin{align} \label{eq:identities}
\g(\bt) = \BB \sin(\B^\top \bt) \quad \text{and} \quad \J(\bt) = \BB \D_\bt \B^\top
\end{align}
where $\D_\bt$ denotes the $|E| \times |E|$ diagonal matrix with diagonal entries given by the vector $\cos(\B^\top \bt)$ and where both sine and cosine act on a vector component wise.

At this point we briefly outline our main result as a means to motivate the remaining definitions and lemmas.

\begin{proposition} \label{prop:bounds}
Define the sets
\begin{align}
\Omega_\up := \BB [-1,1]^{|E|} \quad \text{and} \quad \Omega_\low := \P \g(\Theta_\low)
\end{align}
where
\begin{align}
\Theta_\low := \{ \bt \in [0,2\pi)^N : |\theta_i - \theta_j| < \pi/2 \text{ for all } i,j \in V \}.
\end{align}
If $\hat{G}$ has at least one directed spanning tree, then the volumes satisfying
\begin{align}
\vol_\up = |\P \Omega_\up| \quad \text{and} \quad \vol_\up' \ge \int_{\P_1 [0,2\pi)^{N-1}} |\det( \P \J(\P^\top \bp)\P^\top)| d\bp
\end{align}
and
\begin{align}
\vol_\low = |\P \Omega_\low| = \int_{\P \Theta_\low} |\det( \P \J(\P^\top \bp)\P^\top)| d\bp
\end{align}
satisfy \eqref{eq:inequality}. (Here $\P_i$ denotes $\P$ without the $i$th column.)
\end{proposition}

\begin{proof}
We first establish the inequalities
\begin{align*}
\vol_\low \le \vol_\stab \le \vol_\lock \le \vol_\low
\end{align*}
by proving the set containments
\begin{align*}
\Omega_\low \subseteq \Omega_\stab \subseteq \Omega_\lock \subseteq \Omega_\up.
\end{align*}
The set containment $\Omega_\stab \subseteq \Omega_\lock$ follows from Definition \ref{def:sets} while $\Omega_\lock \subseteq \Omega_\lock$ follows from \eqref{eq:identities} and the observation that $\sin x \in [-1,1]$. Therefore it remains to verify the contianment $\Omega_\low \subseteq \Omega_\stab$ which will follow from $\Theta_\low \subseteq \Theta_\low$. In other words we need to show that $\J(\bt) >  0$ on $\1^\perp$ for $\bt \in \Theta_\low$. To see this note that $\cos(\theta_i - \theta_j) > 0$ for $\bt \in \Theta_\low$ so that the Gershgorin Circle Theorem implies that the real parts of all eigenvalues of $\J(\bt)$ are non-negative. Furthermore, we have that the only possible eigenvalue with zero real part is zero. Therefore it remains to show that $\1$ spans the kernel of $\J(\bt)$. This however is equivalent to showing that the matrix $\J(\bt) \P^\top$ has a trivial kernel which holds if the matrix $\P \J(\bt) \P^\top$ is invertible. This is the case by the Matrix Tree Theorem (Theorem \ref{thm:matrixtree}) since $\hat{G}$ has at least one directed spanning tree.

Next we verify the formula for $\vol_\low$ and extend it to show that $\vol_\lock \le \vol_\up'$. The basic idea is to replace an integral over ``frequencies" with an integral over ``angles". To do this define the function $\tilde{\g}(\P\bt) = \P \g(\bt)$, then its Jacobian $\tilde{\J}(\P \bt) = \P \J(\bt) \P^\top$. By the previous paragraph we know that $\tilde{\J}(\bp) > 0$ for $\bp \in \P \Theta_\low$. Therefore
\begin{align*}
\vol_\low = |\tilde{\g}(\P \Theta_\low)| = \int_{\P \Theta_\low} |\det(\tilde{\J}(\bp))| d\bp = \int_{\P \Theta_\low} |\det(\P \J(\P^\top \bp) \P^\top)| d\bp.
\end{align*}
The inequality for $\vol_\up'$ follows by similar reasoning from $\tilde{\g}(\P_1 [0,2\pi)^{N-1}) = \P \g([0,2\pi)^N)$. This of course simply follows by translating $\bt$ by a multiple of $\1$ so that $\theta_1 = 0$.
\end{proof}

\begin{remark}
Geometrically, $\vol_\up$ and $\vol_\low$ are the $(N-1)$-dimensional Lebesgue measures of the projections of $\Omega_\up$ and $\Omega_\low$ onto the mean zero hyperplane. Unfortunately, $\vol_\up'$ doesn't seem to have such a geometric interpretation.
\end{remark}

In order to obtain our formulas for $\vol_\up'$ and $\vol_\low$ we will need the well known Matrix Tree Theorem (Theorem \ref{thm:matrixtree}) to evaluate the integrals in Proposition \ref{prop:bounds}. Furthermore, by a theorem of Shephard in \cite{MR0362054} we know that $\vol_\up = |\P \BB [-1,1]^{|E|}|$ can be computed in terms of the sub determinants of $\P \BB$ hence the sub determinants of $\BB$. Therefore we characterize the sub determinants of $\BB$ in Lemma \ref{lem:H}.

\begin{theorem}[Matrix Tree Theorem] \label{thm:matrixtree}
For any directed graph $\G$, define the $|V| \times |V|$ graph Laplacian matrix $\L$ by
\begin{align*}
L_{ij} =
\begin{cases}
\sum_{k \ne i} \gamma_{ik} & \mbox{if $i=j$,}
\\
-\gamma_{ij} & \mbox{if $i \ne j$.}
\end{cases}
\end{align*}
Then the minors of $\L$ are sums over directed spanning trees, namely,
\begin{align*}
(-1)^{i+j} \det(\L_{ij}) = \sum_{\T \in \SS\TT_i} \gamma(\T).
\end{align*}
\end{theorem}

\begin{lemma} \label{lem:H}
Let $S$ be a subset of $\E$ with $N-1$ edges. Then,
\begin{align*}
|\det(\BB_{i,S})| =
\begin{cases}
\gamma(\Hh) & \mbox{if there exists an $\Hh \in \HH_i$ such that $E_H = S$},
\\
0 & \mbox{otherwise}.
\end{cases}
\end{align*}
Furthermore, the sign of $\det(\BB_{i,S})$ alternate in $i$.
\end{lemma}

We defer the proof to the appendix but demonstrate the lemma using the subgraphs in Figure \ref{fig:subgraphs}. If we choose $S = \{e_1,e_2,e_4,e_5,e_6\}$, then the left subgraph $\Hh$ in Figure \ref{fig:subgraphs} is an element of $\HH_1$ with $E_H = S$. By inspecting $\G$ in Figure \ref{fig:graphs} we also see that there exists elements of $\HH_4$ and $\HH_5$ satisfying $E_H = S$ but not for $\HH_2$, $\HH_3$, or $\HH_6$. Similarly, if we choose $S' = \{e_1,e_2,e_3,e_6,e_7\}$, then the right subgraph $\Hh$ in Figure \ref{fig:subgraphs} is an element of $\HH_3$ with $E_H = S'$. Again, we find that there exists an element of $\HH_2$ satisfying $E_H = S'$ but not for $\HH_1$, $\HH_4$, $\HH_5$, or $\HH_6$. This can be seen in the computation below.
\begin{align*}
\begin{split}
\det(\BB_{1,S}) &= -\gamma_{54} \gamma_{41} \gamma_{23} \gamma_{36} \gamma_{62}
\\
\det(\BB_{2,S}) &= 0
\\
\det(\BB_{3,S}) &= 0
\\
\det(\BB_{4,S}) &= \gamma_{54} \gamma_{14} \gamma_{23} \gamma_{36} \gamma_{62}
\\
\det(\BB_{5,S}) &= -\gamma_{14} \gamma_{45} \gamma_{23} \gamma_{36} \gamma_{62}
\\
\det(\BB_{6,S}) &= 0
\end{split}
\begin{split}
\det(\BB_{1,S'}) &= 0
\\
\det(\BB_{2,S'}) &= \gamma_{32}(\gamma_{14} \gamma_{45} \gamma_{51}-\gamma_{15} \gamma_{54} \gamma_{41}) \gamma_{65}
\\
\det(\BB_{3,S'}) &= -\gamma_{23}(\gamma_{14} \gamma_{45} \gamma_{51}-\gamma_{15} \gamma_{54} \gamma_{41}) \gamma_{65}
\\
\det(\BB_{4,S'}) &= 0
\\
\det(\BB_{5,S'}) &= 0
\\
\det(\BB_{6,S'}) &= 0
\end{split}
\end{align*}
Further notice how the sign pattern alternates in $i$.

\begin{figure}[H]
\captionsetup{font=scriptsize}
\centering
\begin{subfigure}{.45\textwidth}
  \centering
  \includegraphics[width=.8\linewidth]{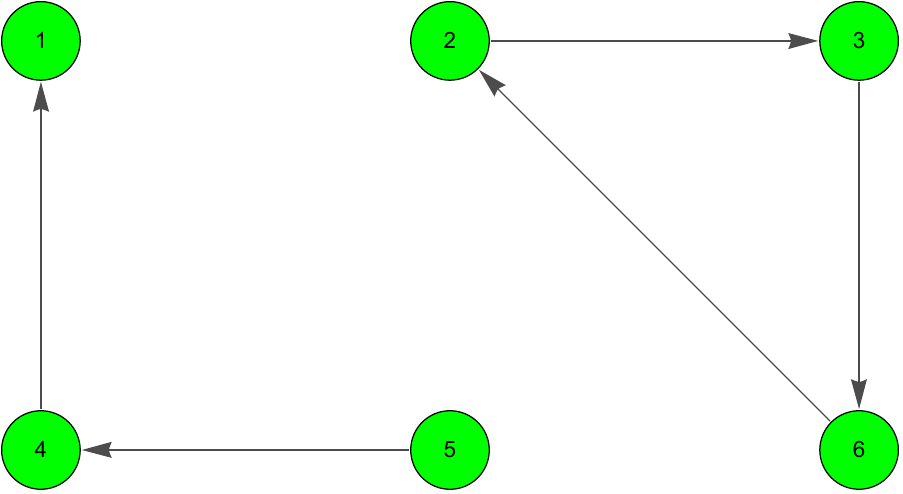}
\end{subfigure}
\begin{subfigure}{.45\textwidth}
  \centering
  \includegraphics[width=0.8\linewidth]{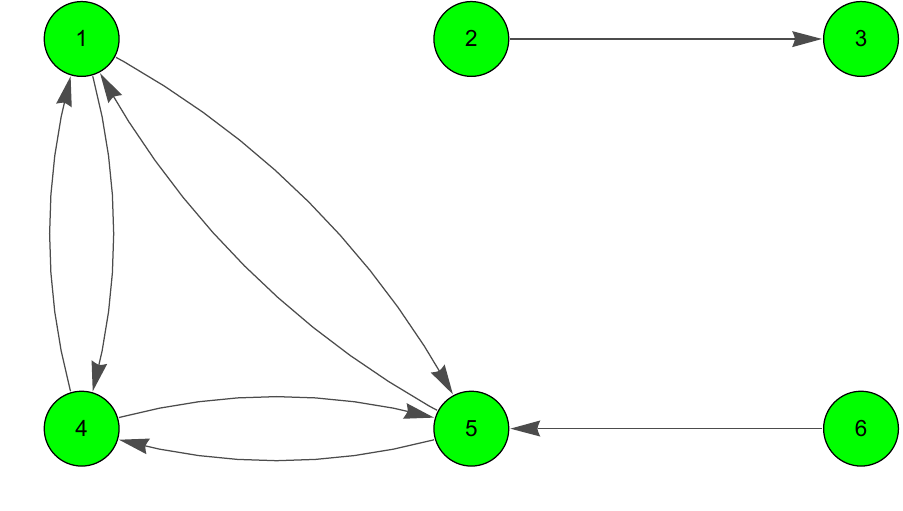}
\end{subfigure}
\caption{Example of subgraphs in $\HH_1$ and $\HH_3$ respectively for the graphs in Figure \ref{fig:graphs}.}
\label{fig:subgraphs}
\end{figure}

We end this section by observing that in the symmetric case, $\gamma_{ij} = \gamma_{ji}$, that $\gamma(\Hh) = 0$ unless $\Hh \in \ST$. To see this it suffices to observe that $\gamma(\R) = 0$ for any ring $\R$ as in Definition \ref{def:H}. Therefore $\Hh$ must be a directed tree and in fact a directed spanning tree. In this paper we will express our upper bound as sum over $\HH$ and this observation explains why our similarly defined upper bound in \cite{doi:10.1137/16M110335X} is a sum over spanning trees.




\section{Proof of Main Theorem} \label{sec:bounds}

In this section we prove Theorem \ref{thm:main}.


\subsection{$\vol_\up$} \label{subsec:upper}

We start by proving our formula for $\vol_\up$ which follows the corresponding proof in \cite{doi:10.1137/16M110335X}. By \cite{MR0362054} and the Cauchy-Binet theorem we know that
\begin{align*}
|\P \Omega_\up| = 2^{N-1} \sum_{\substack{S \subseteq E \\ |S| = N-1}} |\det((\P \BB)_S)| = 2^{N-1} \sum_{\substack{S \subseteq E \\ |S| = N-1}} \biggr\rvert \sum_{i=1}^N \det(\P_i) \det(\BB_{i,S}) \biggr\rvert.
\end{align*}
Since $\P$ is an orthonormal projection onto $\1^\perp$ we have know $\det(\P_i) = \pm (-1)^{i-1}/\sqrt{N}$. Furthermore by Lemma \ref{lem:H} we know that $\det(\BB_{i,S})$ alternates sign in $i$ and has magnitude $\gamma(\Hh)$ if $\Hh \in \HH_i$ with $E_H = S$ and zero otherwise. This gives us the result.


\begin{remark}
At the end of Section \ref{sec:notation} we mentioned that in the symmetric case $\gamma(H) = 0$ unless $H \in \SS \TT$. Furthermore, due to symmetry each spanning tree of $\G$ results in $N$ directed spanning trees of $G$ which allows us to recover our result in \cite{doi:10.1137/16M110335X}, namely,
\begin{align*}
\vol_\up = 2^{N-1} \sqrt{N} \sum_{T \in \mathcal{ST}} \gamma(T).
\end{align*}
\end{remark}



\subsection{$\vol_\up'$ and $\vol_\low$} \label{subsec:upper}

In this section we obtain formulas for $\vol_\up'$ first and then $\vol_\low$. By Proposition \ref{prop:bounds} and the Cauchy-Binet formula we have that
\begin{align*}
\int_{\P_1 [0,2\pi)^{N-1}} |\det( \P \J(\P^\top \bp)\P^\top)| d\bp &= \frac{1}{\sqrt{N}} \int_{[0,2\pi)^{N-1}} |\det( \P \J(\bt)\P^\top)| \biggr\rvert_{\theta_1=0} d\bt
\\
&\le \frac{1}{\sqrt{N}} \int_{[0,2\pi)^{N-1}} \sum_{\T \in \ST} \gamma(\T) \biggr\rvert \prod_{e \in \E_{\T}} \cos \theta_e \biggr\rvert_{\theta_1 = 0} d\bt
\\
&\le \frac{(2\pi)^{N-1}}{\sqrt{N}} \sum_{\T \in \ST} \gamma(\T).
\end{align*}

For $\vol_\low$ we define the sets $\Theta_i = \{ \bt \in [0,\pi/2)^n : \theta_i = 0 \}$ and conclude that
\begin{align*}
\P \Theta_\low = \bigcup_{i=1}^N \P \Theta_i \quad \text{and} \quad |\P  \g(\Theta_i) \cap \P \g(\Theta_j)| = 0 \quad \text{if} \quad i \ne j.
\end{align*}
Therefore from the identity $\P \g(\Theta_i) = \tilde{\g}(\P_i [0,\pi/2]^{N-1})$ we see that it suffices to show that
\begin{align*}
\int_{\P_i [0,2\pi)^{N-1}} |\det( \P \J(\P^\top \bp)\P^\top)| d\bp = \frac{1}{\sqrt{N}} \sum_{\T \in \ST} \gamma(\T) \prod_{j \ne i} I(\deg_j(\T)).
\end{align*}
This however follows from the Cauchy-Binet formula, the Matrix Tree Theorem, and Lemma 8 in \cite{doi:10.1137/16M110335X} since
\begin{align*}
\int_{\P_i [0,2\pi)^{N-1}} |\det( \P \J(\P^\top \bp)\P^\top)| d\bp &= \frac{1}{\sqrt{N}} \int_{[0,\pi/2)^{N-1}} |\det(\P \J(\bt) \P^\top)|\biggr\rvert_{\theta_i=0} d\bt
\\
&= \frac{1}{\sqrt{N}} \int_{[0,\pi/2)^{N-1}}  \sum_{\T \in \ST} \gamma(\T) \prod_{e \in \E_{\T}} \cos \theta_e \biggr\rvert_{\theta_i = 0} d\bt
\\
&= \frac{1}{\sqrt{N}} \sum_{\T \in \ST} \gamma(\T) \prod_{e \in \E_{\T}} \int_{[0,\pi/2)^{N-1}} \cos \theta_e \biggr\rvert_{\theta_i=0} d\bt.
\end{align*}


\begin{remark}
Again in the symmetric case every spanning tree of $\G$ results in $N$ spanning trees of $G$ all with the same degree vector $\deg(T) = \deg(\T)$. Therefore we again recover our result from \cite{doi:10.1137/16M110335X},
\begin{align*}
\vol_\low = \sqrt{N} \sum_{T \in \mathcal{ST}} \gamma(T) \sum_{i=1}^N \prod_{j \ne i} I(\deg_j(T)).
\end{align*}
\end{remark}


\section{Comparisons} \label{sec:comp}

In this section we start by comparing results for the symmetric case in \cite{doi:10.1137/16M110335X} and asymmetric case here. Then we compare our two upper bounds $\vol_\up$ and $\vol_\up'$. First note that the lower bound in Theorem \ref{thm:main} and in \cite{doi:10.1137/16M110335X} are both sums over spanning trees and that the trees contribute in the same way. Therefore Theorem 13 in \cite{doi:10.1137/16M110335X}, which shows that $\vol_\low$ is Schur-convex on trees with respect to the degree vector $\deg(T)$, and its consequences immediately translate to the asymmetric case. In particular, we have upper and lower bounds for $\vol_\lock$ and $\vol_\stab$ in terms of the number of directed spanning trees.

Another theorem that partially translates from the symmetric case to the asymmetric case is Theorem 18 in \cite{doi:10.1137/16M110335X} which states that for dense networks that the logarithms of the volumes $\vol_\lock$ and $\vol_\stab$ are asymptotic to the logartithm of the number of spanning trees. In the asymmetric case we get the following theorem.

\begin{theorem} \label{thm:dense}
Let $\G_N$ be a family of networks for which $\# \ST_{\G_N}$ has super-exponential growth. Then
\begin{align}
\lim_{N \rightarrow \infty} \frac{\log \vol_\lock(\G_N)}{\log \# \ST_{\G_N}} = \lim_{N \rightarrow \infty} \frac{\log \vol_\stab(\G_N)}{\log \# \ST_{\G_N}} = 1.
\end{align}
(Here, as in \cite{doi:10.1137/16M110335X}, all edge weights are taken to be one.)
\end{theorem}

The proof of Theorem \ref{thm:dense} is the same as in \cite{doi:10.1137/16M110335X} and uses the fact that each spanning tree contributes at most an exponential factor in $N$. The only difference is that we don't have an analog of \cite{Bogdanowicz.2009} which allows us to conclude that the number of spanning trees grows super-exponentially for ``dense" networks.

Now we compare our two upper bounds $\vol_\up$ and $\vol_\up'$. They are both sums over subgraphs of $\G$, but $\vol_\up$ is a sum over a larger set than $\vol_\up'$. However, $\vol_\up'$ has a larger exponential factor in front. Therefore $\vol_\up$ is probably better for sparse networks, in particular directed trees. Also, as was already observed, if the edge weights are symmetric, $\gamma_{ij} = \gamma_{ji}$, then $\gamma(\Hh) = 0$ unless $\Hh \in \ST$. Therefore if the edge weights are nearly symmetric we expect that $\vol_\up$ will again be smaller than $\vol_\up'$ by the exponential factor. On the other hand computing the number of spanning trees $\# \ST$ is already computationally demanding and $\# \HH$ is even worse since $\HH$ contains $\ST$. Also, even though computing $\# \ST$ is a difficult problem, it has been studied in the literature as mentioned in the introduction. Therefore, from a computational perspective $\vol_\up'$ is a more manageable quantity.


\section{Motifs} \label{sec:motifs}

In \cite{10.3389/fncom.2011.00028} Nykamp etal. numerically investigated the dependence of synchronization models on the statistical properties of certain two edge subgraphs called motifs. Since Theorem \ref{thm:main} demonstrates that synchronization is closely related to the complexity of a network, the number of spanning trees, we seek to determine if the complexity of a network depends on the presence of these motifs in the same way as found by Nykamp etal.

In directed networks there are four basic motifs, referred to as reciprocal, convergent, divergent, and chain, which are displayed in Figure \ref{fig:motifs}. The reciprocal motif consists of two edges with the same two vertices but opposite orientation. The convergent, divergent, and chain motifs each consist of two edges sharing a single vertex. In the convergent motif both edges point towards the common vertex while they point away from the common vertex in the divergent motif. Finally, in the chain motif one edge points towards the common vertex while the other edge points away.

\begin{figure}[H]
\centering
\begin{tabular}{cccc}
\includegraphics[width=37mm]{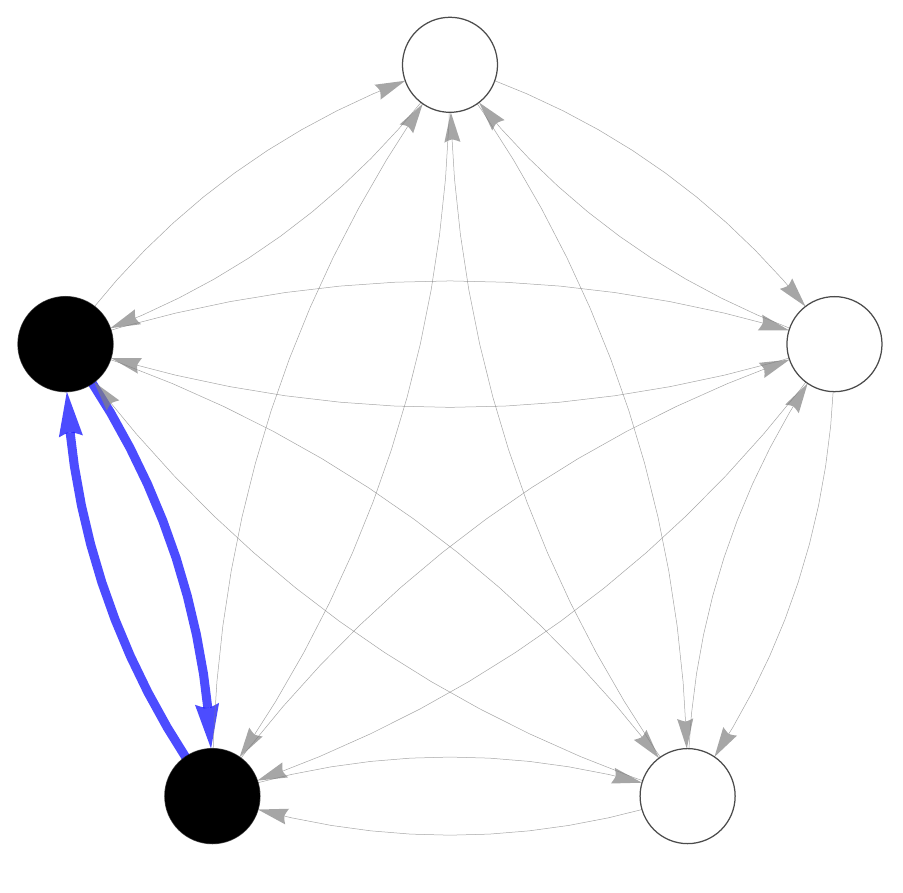} & \includegraphics[width=37mm]{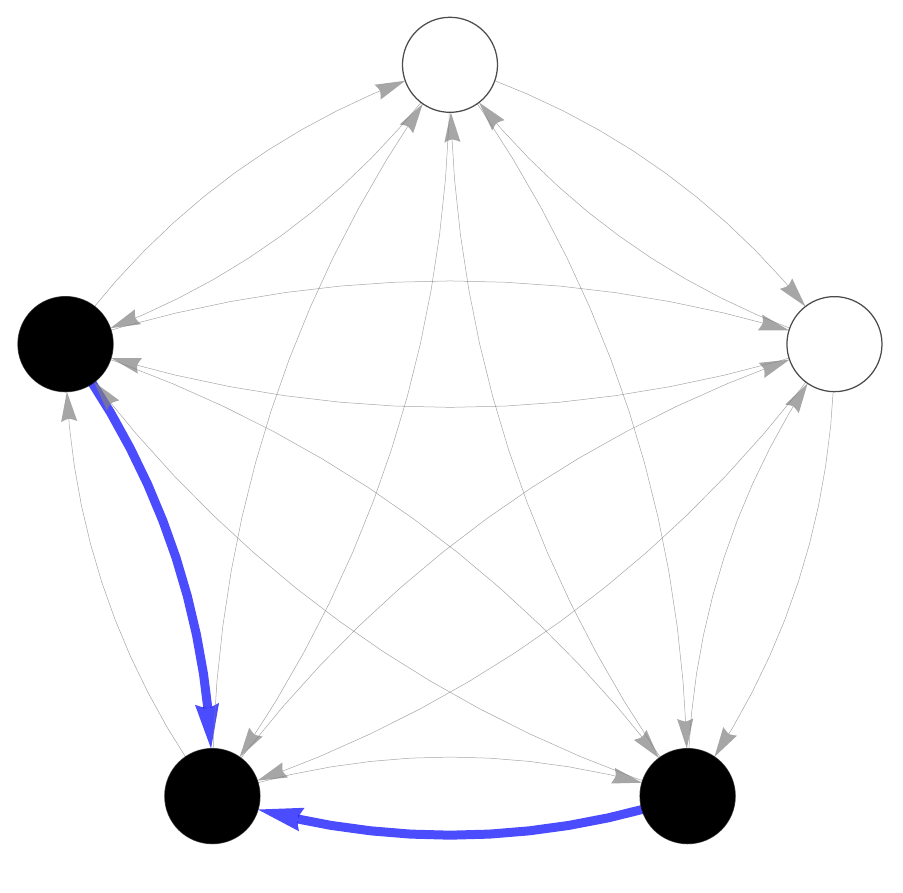} & \includegraphics[width=37mm]{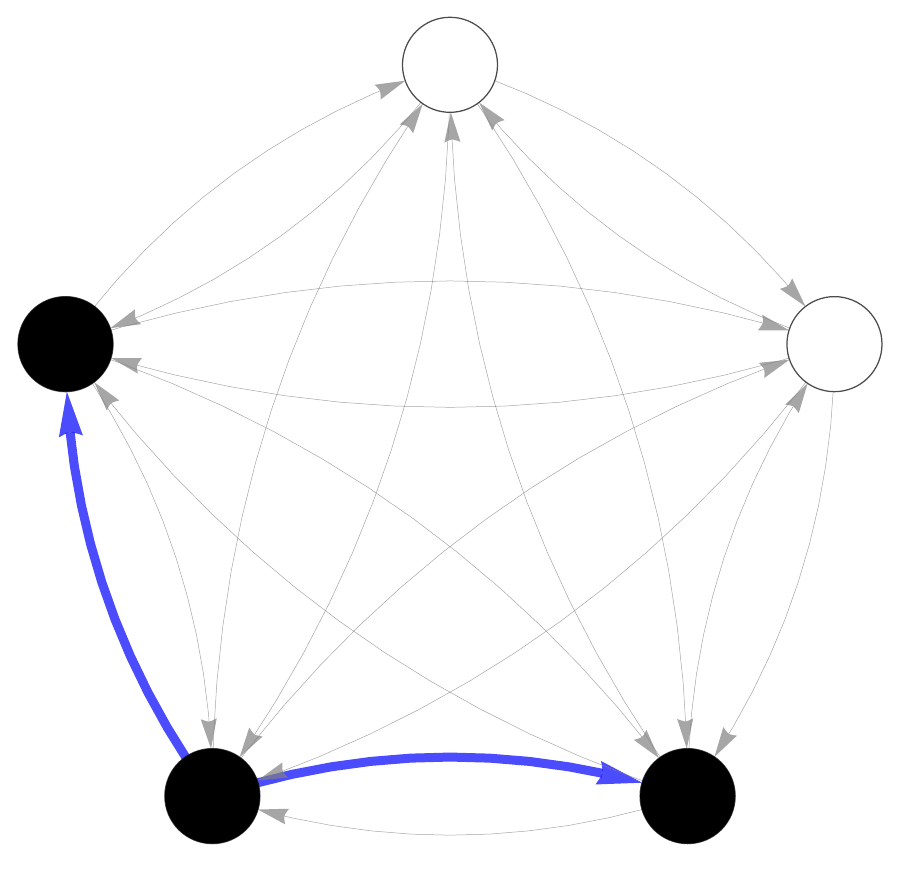} & \includegraphics[width=37mm]{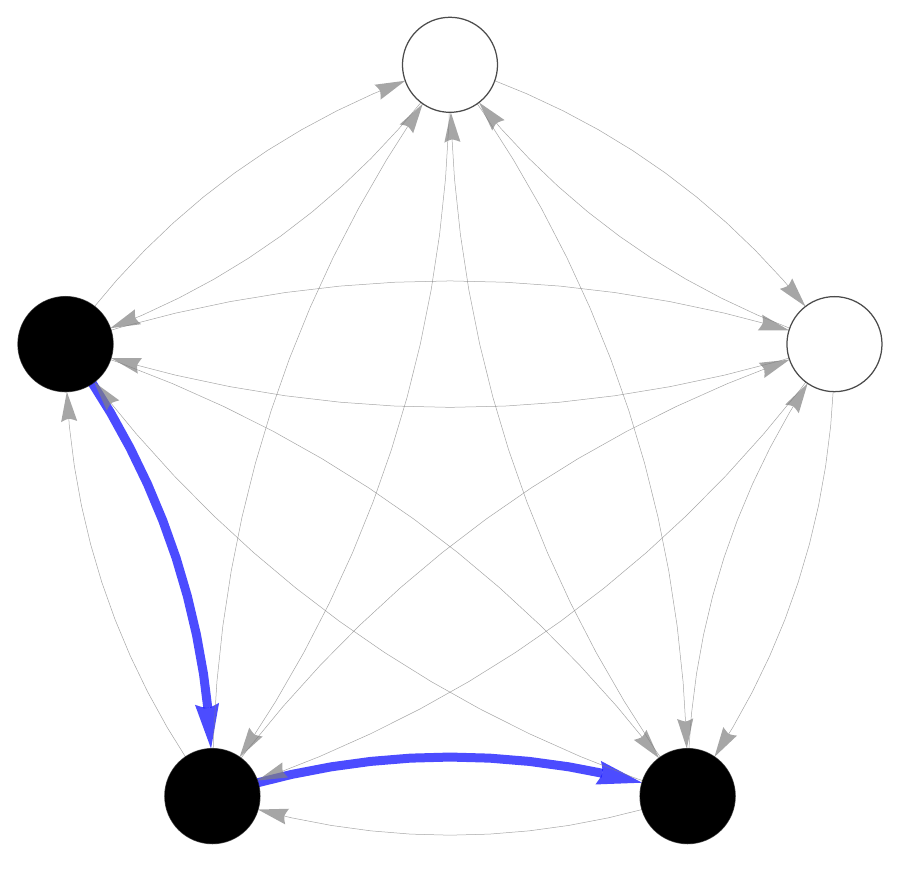}
\end{tabular}
\caption{The reciprocal, convergent, divergent, and chain motifs respectively.}
\label{fig:motifs} 
\end{figure}

Nykamp etal. measured the affect of these motifs on synchronization by evolving random initial data until it reached a steady state $\bt$ and then computing the order parameter which is defined to be the magnitude of the exponential sum
\begin{align} \label{eq:order}
r = \frac{1}{N} \biggr\rvert \sum_{j=1}^N e^{i \theta_j} \biggr\rvert
\end{align}
where $i = \sqrt{-1}$. The larger the value of $r$ the more synchronized the solution is. For example, complete synchronization, all angles are identical, is occurs if and only if $r$ is its maximum value $1$. An example when $r$ achieves its minimum value of $0$ is when the angles are uniformly distributed on a ring network. They then defined the observed statistics
\begin{gather*}
\hat{p} = \frac{N_{\text{conn}}}{N(N-1)},
\\
\hat{p}^2(1+\hat{\alpha}_{\text{recip}}) = \frac{N_{\text{recip}}}{N(N-1)/2},
\\
\hat{p}^2(1+\hat{\alpha}_{\text{conv}}) = \frac{N_{\text{conv}}}{N(N-1)(N-2)/2},
\\
\hat{p}^2(1+\hat{\alpha}_{\text{div}}) = \frac{N_{\text{div}}}{N(N-1)(N-2)/2},
\\
\hat{p}^2(1+\hat{\alpha}_{\text{chain}}) = \frac{N_{\text{chain}}}{N(N-1)(N-2)}.
\end{gather*}
where $N_{\text{conn}}$, $N_{\text{recip}}$, $N_{\text{conv}}$, $N_{\text{div}}$, $N_{\text{chain}}$ denote the number of edges and respective motifs in the network. They found that synchronization is not significantly affected by $\hat{\alpha}_{\text{recip}}$ or $\hat{\alpha}_{\text{div}}$, but that it increases as $\hat{\alpha}_{\text{chain}}$ increases and $\hat{\alpha}_{\text{div}}$ decreases. As a result we seek to numerically test if the complexity of the network depends on these five observed statistics in the same way.

Nykamp, Zhao etal. \cite{10.3389/fncom.2011.00028,Zhao.2012} discuss a method of generating a class of random networks which are a natural generalization the Erd\H{o}s-R\'{e}nyi random networks using these motifs. To construct an Erd\H{o}s-R\'{e}nyi random network we begin with a vertex set and then include every possible edge independently with probability $p$. They introduced correlations between these edges in order to change the probability of the two edge motifs from $p^2$ which is its value in the uncorrelated case. They did this by generating a vector of independent standard normal random variables and then acting on this vector by an appropriate matrix to derive a vector of correlated standard normal random variables. Finally, they threshold these random variables to determine which edges to include. This method of generating random networks is also discussed by Bronski and the author \cite{2018arXiv180805076B} where we use the algebraic structure of a coherent configuration to simplify computations.

\begin{figure}
\centering
\begin{tabular}{ccc}
\includegraphics[width=5.5cm,height=4cm,keepaspectratio]{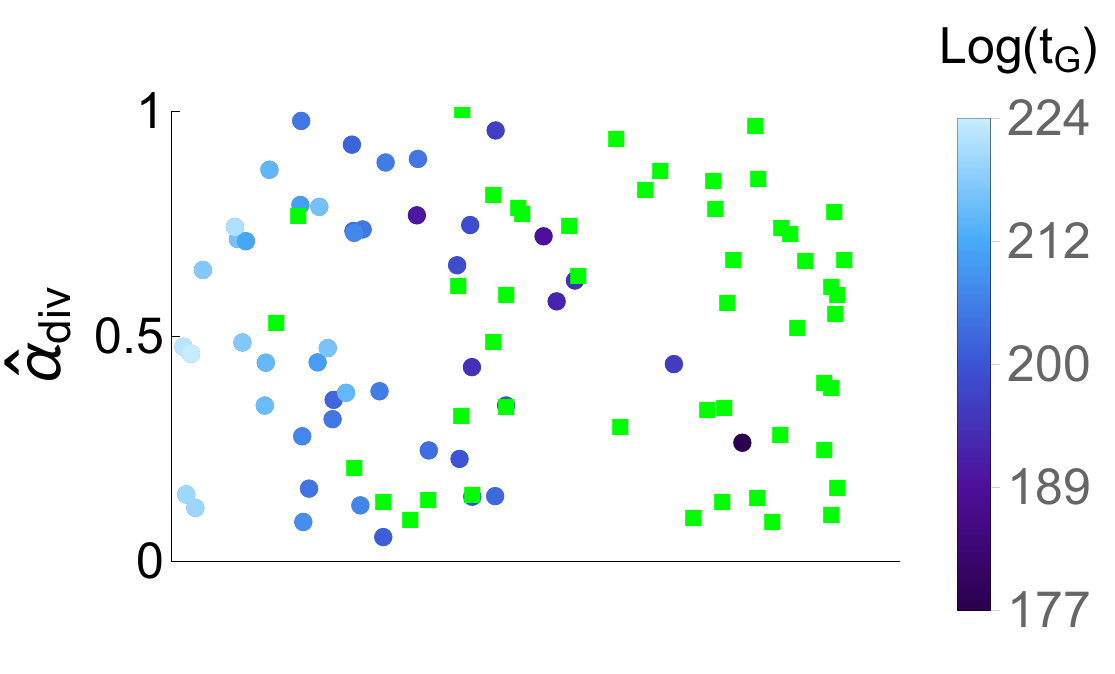}
\\
\includegraphics[width=5.5cm,height=4cm,keepaspectratio]{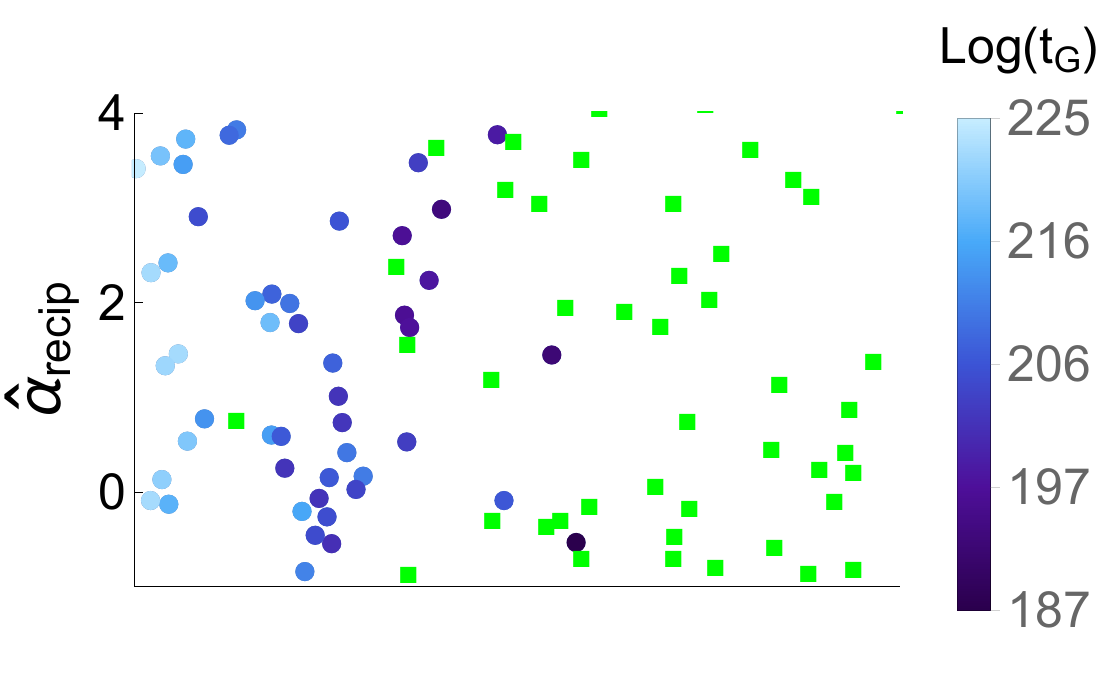} & \includegraphics[width=5.5cm,height=4cm,keepaspectratio]{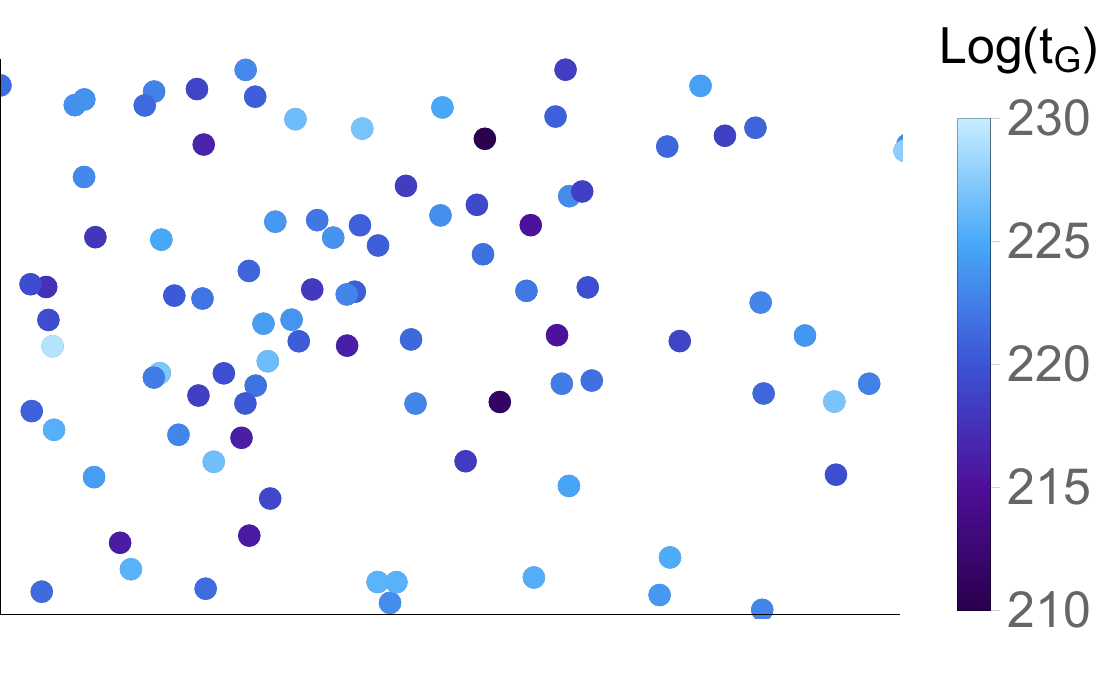}
\\
\includegraphics[width=5.5cm,height=4cm,keepaspectratio]{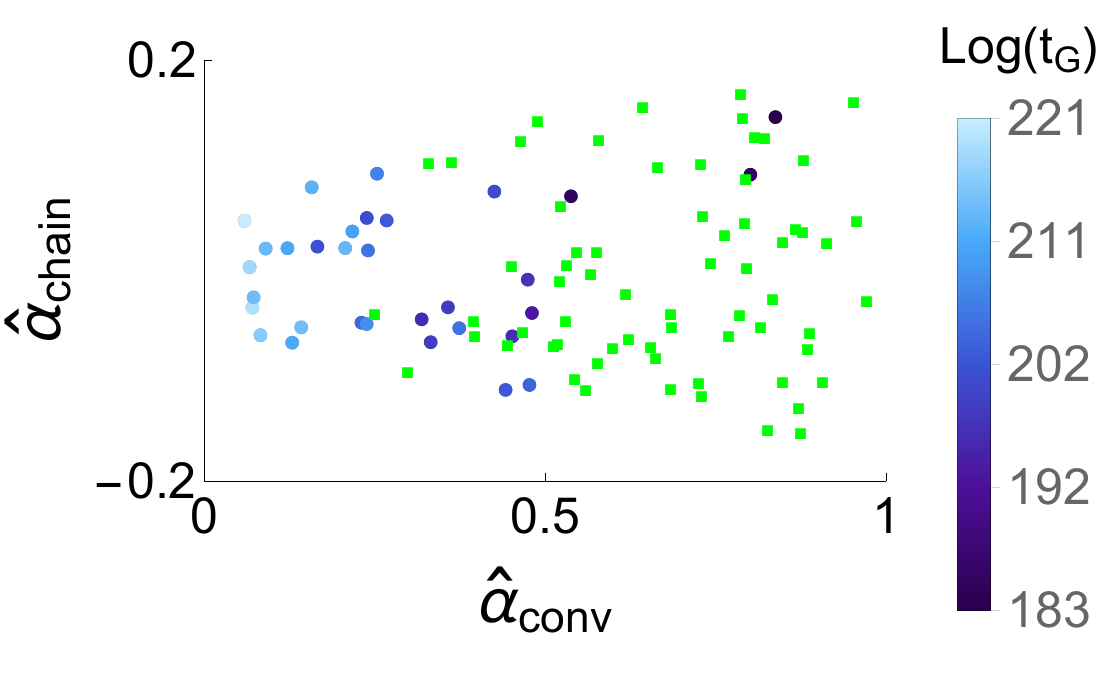} & \includegraphics[width=5.5cm,height=4cm,keepaspectratio]{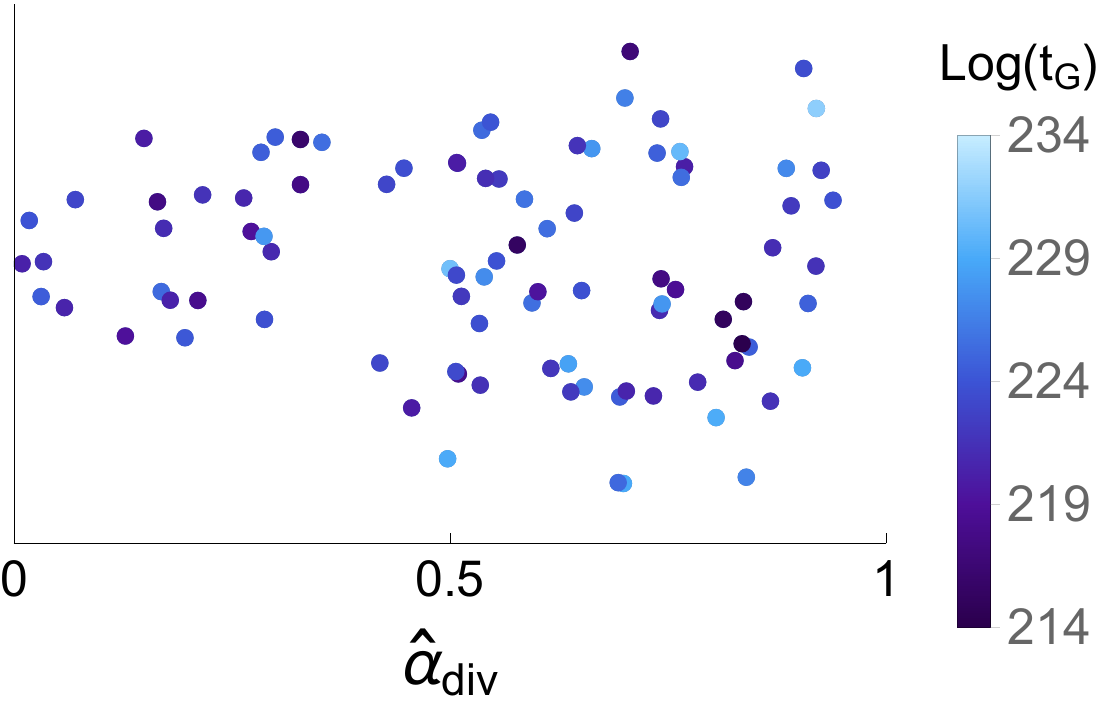} & \includegraphics[width=5.5cm,height=4cm,keepaspectratio]{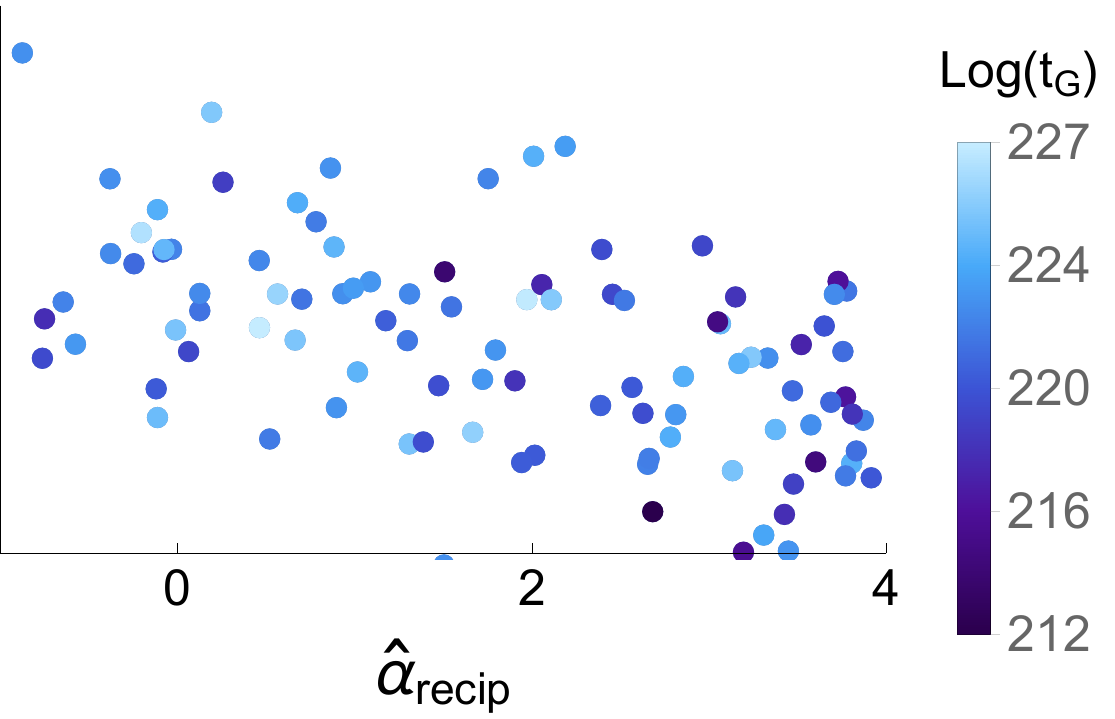}
\end{tabular}
\caption{Plots numerically simulating how the number of directed spanning trees depends on each pair of the observed statistics $\hat{\alpha}_{recip}$, $\hat{\alpha}_{conv}$, $\hat{\alpha}_{div}$, and $\hat{\alpha}_{chain}$. Each plot contains data from 100 random graphs each with $N = 100$ vertices and $p = 0.1$ (this is the theoretical probability of an edge being included). A green data point represents a graph without a directed spanning tree while the legend displays how the color depends on the logarithm of the number of directed spanning trees when they exist. In each plot the $\hat{\alpha}$'s for the two motifs not being compared are nearly zero (there theoretical value is zero).}
\label{fig:data}
\end{figure}

We plot the complexity of random networks verses the observed statistics in Figure \ref{fig:data}. It appears that the complexity decreases as $\hat{\alpha}_{conv}$ inreases and is relatively unaffected by $\hat{\alpha}_{recip}$, $\hat{\alpha}_{div}$, and $\hat{\alpha}_{chain}$. This is in agreement with \cite{10.3389/fncom.2011.00028} for all but the chain motif which was found by in \cite{10.3389/fncom.2011.00028} to result in an increase in synchrony.

\section{conclusion}

In this paper we generalized our results in \cite{doi:10.1137/16M110335X} to the asymmetric case. We found upper and lower bounds for the projection of the synchronization region onto the mean zero hyperplane. Our lower bound is a sum over directed spanning trees while our upper bound is a sum over a class of directed subgraphs containing directed spanning trees. We then used numerical simulations to determine the dependence of the number of directed spanning trees of a graph on the presence of four two edge motifs in the network. We found that this dependence was similar to the dependence of synchronization on these motifs in \cite{10.3389/fncom.2011.00028} with the exception of the chain motif.



\section{Acknowledgements}

The author gratefully acknowledges support under NSF grant DMS1615418.



\section{Appendix}

\begin{proof}[Proof of Lemma \ref{lem:H}]
First suppose that $\hat{K} = \T \sqcup \hat{J}$ is a subgraph of $\G$ such that $E_K = S$ and $K = T \sqcup J$ where $T$ is a tree. We show that if $i$ is not the root of $\T$, then $\det(\BB_{i,S}) = 0$. First if $\T$ is a single vertex, then the entire $i$th is zero hence the result. Now if $\T$ has more than one vertex we can choose a leaf of $\T$ and expand the determinant along the column which represents the edge that connects it to the rest of $\T$. If $i$ is not connected to another vertex in $\T$ by an edge in $G$ directed from $i$ to the other vertex, then this results in a determinant of zero. If it does, then we get the weight of the edge multiplied by the same determinant but for a graph with this edge and leaf removed. We can continue this process until $\T$ is reduced to a single vertex in which case we are done.

Therefore since $S$ contains $N-1$ edges and can not represent a subgraph containing two or more disjoint trees, it must therefore represent a subgraph of $\G$ of the form
\begin{align*}
\Hh = \T \sqcup \bigsqcup_{\R} \R \bowtie_r \T_r,
\end{align*}
Recall that this notation represents the disjoint union of a tree $\T$ and subgraphs with a single cycle, namely, a ring $\R$ with trees $\T_r$ attached to $\R$ at the vertices $r$. From our above argument of expanding along columns with edges connecting leaves we know that the determinant is zero unless $i$ is the root of $\T$ and $r$ is the root of $\T_r$. Also, the contribution of the trees $\T$ and $\T_r$ are the weights $\gamma(\T)$ and $\gamma(\T_r)$.

Therefore it remains to consider the sub determinants of $\BB$ with edges $\E_{\R}$ forming a single cycle in $G$. For simplicity we label the vertices in $\R$ cyclically as $1,\dots,\ell$ and assume that $\R$ is a double directed. In this case our sub determinant becomes
\begin{align*}
|\det
\begin{pmatrix}
-\gamma_{12} & 0 & 0 & \hdots & 0 & 0 & -\gamma_{1\ell}
\\
\gamma_{21} & -\gamma_{23} & 0 & \hdots & 0 & 0 & 0
\\
0 & \gamma_{32} & -\gamma_{34} & \hdots & 0 & 0 & 0
\\
\vdots & \vdots & \vdots & \ddots & \vdots & \vdots & \vdots
\\
0 & 0 & 0 & \hdots & -\gamma_{\ell-2 \ell-1} & 0 & 0
\\
0 & 0 & 0 & \hdots & \gamma_{\ell-1 \ell-2} & -\gamma_{\ell-1 \ell} & 0
\\
0 & 0 & 0 & \hdots & 0 & \gamma_{\ell \ell-1} & \gamma_{\ell 1}
\end{pmatrix}|
=
|\gamma_{1\ell} \gamma_{\ell \ell-1} \dots \gamma_{32} \gamma_{21} - \gamma_{12} \gamma_{23} \dots \gamma_{\ell-1 \ell} \gamma_{\ell 1}|
\end{align*}
by expanding the along the last column. By then setting edge weights equal to zero for edges not contained in $\G$ we see that the determinant is zero unless $\R$ is a single directed cycle or double directed cycle as described in Definition \ref{def:H}. This again contributes the weight $\gamma(\R)$ completing the first part of the proof.

Now it remains to show that our determinant alternates sign in $i$. Since we only need to consider the case when $i$ is a vertex of $\T$ we can suppose for simplicity that $\G$ is in fact a directed tree. Now since $\1$ belongs to the null space of $\J(\0) = \BB \B^\top$ we know that there exists a vector $\vv$ in the left null space of $\J(\0)$. But $-\J(\0)$ is negative definite on the orthogonal complement of $\1$, and therefore all of its left eigenvectors have eigenvalues whose real parts are strictly negative. In other words, zero, with left eigenvector $\vv$, is the eigenvalue with the largest real part. Furthermore, all of its off diagonal entries are positive by assumption, and therefore we can add a sufficiently large multiple of the identity matrix to $-\J(\0)$ to obtain a non-negative matrix. This simply shifts the spectrum to the right. But now the shifted zero eigenvalue is the Perron-Frobenius eigenvalue and so we can take $\vv$ to have strictly positive entries by the Perron-Frobenius theorem. Now since the left null space of $\B^\top$ is trivial we know that $\vv^\top \BB = \0$. Let $\rr_k$ for vertices $k$ denote the rows of $\BB$ and let $i$ and $j$ denote any two vertices. We can write $\rr_i = -\sum_{k \ne j} \frac{v_k}{v_i} \rr_k$. This allows us to replace $\rr_i$ with $-\frac{v_j}{v_i} \rr_j$. Note that all other rows can be removed by properties of the determinant. Then swapping rows and factoring out $-\frac{v_j}{v_i}$ shows that the determinants $\det(\BB_{i,S})$ and $\det(\BB_{j,S})$ have signs differing by $(-1)^{i-j}$.
\end{proof}



\bibliography{DirectedKuramotoVolume}
\bibliographystyle{plain}



\end{document}